\documentclass[sigconf, acmcopyright=ifaamas]{aamas}

\makeatletter
\renewcommand{\footnotetextcopyrightpermission}[1]{} 
\makeatother

\makeatletter
\def\@copyrightspace{\relax}
\makeatother

\usepackage{balance} 

\let\cB\relax      
\usepackage{auxlib}
\usepackage{cleveref}      
\usepackage{hyperref}      

\setcopyright{none}

\title{Multiagent Matroid Upgrading: Greedy is Fair and Efficient}

\author{Qingwen Ma}
\orcid{0009-0007-0899-9639}
\affiliation{
  \institution{East China Normal University}
  \department{School of Software Engineering}
  \city{Shanghai}
  \country{China}
}
\email{51265902177@stu.ecnu.edu.cn}

\author{Chao Peng}
\orcid{0009-0006-1695-7693}
\affiliation{
  \institution{East China Normal University}
  \department{School of Software Engineering}
  \city{Shanghai}
  \country{China}
}
\email{cpeng@sei.ecnu.edu.cn}

\author{Changfeng Xu}\thanks{All authors (ordered alphabetically) have equal contributions and are corresponding authors.}
\orcid{0009-0009-2486-5491}
\affiliation{
  \institution{East China Normal University}
  \department{School of Software Engineering}
  \city{Shanghai}
  \country{China}
}
\email{51265902120@stu.ecnu.edu.cn}

\author{Chenyang Xu}
\orcid{0000-0002-6837-2906}
\affiliation{
  \institution{East China Normal University}
  \department{School of Software Engineering}
  \city{Shanghai}
  \country{China}
}
\email{cyxu@sei.ecnu.edu.cn}

\author{Ruilong Zhang}
\orcid{0000-0002-4859-2661}
\affiliation{
  \institution{City University of Hong Kong (Dongguan)}
  \department{Department of Computer Science}
  \city{Dongguan}
  \country{China}
}
\email{ruilong.zhang@cityu-dg.edu.cn}

\begin{abstract}

This paper introduces a general multiagent matroid upgrading problem that models a broad class of real-world resource allocation tasks. In this setting, there are multiple agents and a ground set of elements, where each element is assigned to a specific agent and has two associated costs: a default cost and a reduced (upgraded) cost. Upgrading an element lowers its cost to the upgraded value, while non-upgraded elements retain their default costs. Each agent is associated with its own matroid, with the goal of finding a minimum-cost basis. The central task is to select at most $k$ elements to upgrade so as to minimize a non-decreasing convex function over the agents' minimum basis costs, capturing both efficiency and fairness objectives in multiagent systems.

We show that the problem is polynomial-time solvable and that an optimal solution can be obtained via a simple greedy algorithm. Our analysis exploits the structural properties of matroids to establish the existence of optimal substructures, thereby ensuring that greedy upgrading yields optimal outcomes. Building on this insight, we can further extend our result to more general settings, such as scenarios with interval fairness constraints, where the number of elements upgraded for each agent is required to lie within a specified interval.
\end{abstract}


\newcommand{\BibTeX}{\rm B\kern-.05em{\sc i\kern-.025em b}\kern-.08em\TeX}

\begin{document}


\pagestyle{fancy}
\fancyhead{}


\maketitle


\section{Introduction}



The allocation of limited resources among multiple agents is a fundamental problem in multiagent systems, with applications ranging from communication networks~\cite{DBLP:journals/access/AhmedFRAQKSC25} and distributed computing~\cite{DBLP:journals/tsc/OustadYASSHE25} to robotics~\cite{DBLP:journals/ras/RathiDBKTTJ22} and transportation~\cite{DBLP:journals/tits/ManogaranGN23}. In such scenarios, resources are typically scarce, and different agents may face their own structural constraints on how resources can be utilized, e.g., maintaining feasibility with respect to connectivity, capacity, or diversity requirements~\cite{DBLP:conf/aaai/AmaniT21}. At the same time, the system usually needs to consider both efficiency, which seeks to maximize overall benefit, and fairness, which aims to balance the outcomes among agents~\cite{DBLP:conf/aaai/GrupenSL22}. These challenges motivate the development of models that integrate efficiency- and fairness-oriented objectives with combinatorial constraints in a multiagent resource allocation framework.

In this paper, we study a general class of \emph{matroid upgrading} problems in multiagent systems, which captures a variety of applications such as network upgrading and neural network compression. In the model, there is a ground set of elements $E$, where each element $e \in E$ is associated with two non-negative costs: a default cost $\dc(e)$ and an upgraded cost $\uc(e)$, with $\dc(e) \geq \uc(e)$. Initially, each element incurs its default cost. However, if selected for upgrading, its cost is reduced to the upgraded value.

The elements are partitioned into $n$ disjoint groups $\{E^{(1)}, \dots, E^{(n)}\}$, with each group assigned to an agent. Each agent $i\in [n]$ is associated with a matroid\footnote{A matroid is a set system $(E,\cI)$ with $\cI\subseteq 2^E$ such that (\rom{1}) $\emptyset\in\cI$; (\rom{2}) for each $S\in\cI$, all $S$'s subsets are also in $\cI$; (\rom{3}) If $A,B\in\cI$ with $\abs{A}<\abs{B}$, then $\exists e\in B\setminus A$ such that $\set{e}\cup A\in\cI$.} $\cM^{(i)} = (E^{(i)}, \cI^{(i)})$, and the objective is to find a minimum-cost basis within the matroid. We can select a subset $S \subseteq E$ of at most $k$ elements to upgrade, with the objective of minimizing
$\sum_{i \in [n]} F_i\bigl(\delta_S(\cM^{(i)})\bigr) $,
where $\delta_S(\cM^{(i)})$ denotes the minimum cost of a basis in $\cM^{(i)}$ given that only elements in $S$ are upgraded, and each $F_i$ is a non-decreasing convex function.
Formally, the multiagent matroid upgrading problem (MMUP) can be written as:
\begin{equation}
\label{eq:MMUP}
\begin{aligned}
    \quad \min_{S \subseteq E,\, |S| \le k} & \sum_{i \in [n]} F_i\left( \delta_S(\cM^{(i)}) \right) \\
    \text{ s.t. }\quad \delta_S(\cM^{(i)}) &= \min_{B_i \in \cB^{(i)}} \sum_{e \in B_i} c_S(e) \;\;\;\forall i \in [n]~,
\end{aligned}
\tag{MMUP}
\end{equation}
where $\cB^{(i)}$ denotes the set of bases of matroid $\cM^{(i)}$, and the cost $c_S(e)$ is defined as $c_S(e) = \uc(e)$ if $e \in S$, and $c_S(e) = \dc(e)$ otherwise. We remark that the objective function captures a variety of aggregation goals, such as the weighted sum of agents' utilities when overall efficiency is prioritized, or the $\ell_q$ norm of the utilities when fairness across agents is the primary concern.


\subsection{Applications}

This subsection presents two concrete examples in which the agents can be naturally modeled with matroid rank constraints. 


\paragraph{ Resource Allocation for Network Upgrade.} 
Network upgrading plays an important role in ensuring the reliability, efficiency, and scalability of modern infrastructure. As networks continue to evolve and expand, they require periodic maintenance and upgrades to ensure safety and service quality~\cite{natalino2019infrastructure,buys2006road,balakrishnan2009connectivity,yao2011improving,nag2012energy}. However, in contrast to the scale of such networks, the human and operational resources available to the managing organizations are often highly limited. For instance, in~\cite{blueally2020upgrade}, an entertainment provider operates hundreds of retail locations across the United States, Europe, and Asia, all requiring network upgrades, yet its engineering team remains relatively small compared to its global infrastructure footprint. This challenge can be naturally modeled within our framework: each agent corresponds to a retail location, with its utility reflecting the quality of network connectivity, which can be represented by a graphic matroid (e.g., a forest matroid). The coordinator then aims to upgrade selected network links to ensure both high efficiency and fairness across the agents.


\paragraph{Neural Network Compression.} Neural network compression is a popular research area in deep learning, aiming to reduce the computational and storage demands of deep neural networks (DNNs) to facilitate their deployment on resource-constrained embedded systems. Given the substantial computational costs and storage requirements of DNNs, deploying them on such systems is challenging. Consequently, numerous studies have focused on compressing neural networks~\cite{DBLP:journals/corr/HanMD15,DBLP:conf/nips/SerraYKR21,DBLP:conf/nips/CheeFDS22}. However, excessive compression or pruning can adversely affect model accuracy. Therefore, it is essential to limit the extent of compression or pruning to maintain performance. 
This scenario can be formulated within our model, where each channel or block is treated as an agent, and its utility corresponds to the computational cost associated with that component. Notably, this cost can be approximated by a matroid basis cost. Supporting this viewpoint, prior work~\cite{DBLP:conf/iccv/VoTBKH23} modeled binary neural networks as fully connected graphs, where each vertex corresponds to an output channel and the edge weights represent the Hamming distances between associated weight vectors. By constructing a minimum spanning tree (MST) over this graph, they proposed the MST-compression method, which effectively reduces computational cost and latency.





\subsection{Our Contributions}

In this paper, we formalize the \emph{multiagent matroid upgrading problem} (MMUP), which integrates efficiency- and fairness-driven objectives with combinatorial constraints arising from matroid structures. While resource allocation in multiagent systems is generally challenging, we identify a structured subclass grounded in matroid theory that admits efficient and provably optimal solutions via a simple greedy approach.


\begin{mainthm}
    Given any MMUP instance, there exists a general greedy algorithm that computes an optimal solution in polynomial time.
\end{mainthm}

Our main contribution is a polynomial-time greedy algorithm that solves MMUP optimally. The algorithm is natural and simple: we begin by assuming that all elements are upgraded and compute the minimum-cost basis for each agent. The union of all these bases forms a candidate element set. 
Then, we iteratively select the element from this candidate set that results in the largest decrease in the leader's objective and include it in the upgrade set, continuing this process until $k$ elements have been selected.
Although the algorithm is conceptually simple, proving its optimality is non-trivial. The correctness relies on a crucial structured property of optimal solutions, which we refer to as the \emph{nestedness property}. This property ensures that for any optimal solution involving $k-1$ upgrades, there always exists an optimal solution with $k$ upgrades that contains the former as a subset.

By leveraging the nestedness property, we further show that our results extend to several more general settings. For instance, the objective can be replaced with a minimax form to capture worst-case efficiency among agents, i.e., $\max_{i\in [n]} F_i(\delta_S(\cM^{(i)}))$, Alternatively, one can incorporate recently popular interval fairness constraints~\cite{icml/HalabiFNTT23,kdd/Cui00LL24} which require that the number of upgraded elements for each follower must lie within a specified interval. We prove that in both of these generalized models, the greedy algorithm remains optimal.
The details are shown in \cref{sec:exten}.

\paragraph{Connection to Budget-Constrained MST.} 
As a byproduct, we find that our model captures the special case of the budget-constrained minimum spanning tree (MST) problem with $\{0,1\}$ edge weights. This classical problem aims to select a spanning tree of minimum total cost subject to a weight budget constraint. It is known to be NP-hard in general, and existing approaches are based on LP techniques, including Lagrangian relaxation~\cite{DBLP:conf/swat/RaviG96} and LP rounding~\cite{DBLP:journals/mp/GrandoniRSZ14}. In contrast, our algorithm is purely combinatorial and implies that the $\{0,1\}$-weight case can be solved in polynomial time. To the best of our knowledge, this positive result was not known before.
See \cref{sec:connection} for details.

\subsection{Other Related Works}





There are several other upgrade models studied in the literature. For example, the scheduling with testing introduced by~\cite{DBLP:journals/algorithmica/DurrEMM20} is closely related to the upgrading framework considered in this paper. Recent works~\cite{DBLP:conf/waoa/AlbersE20,DBLP:conf/wads/AlbersE21,gong2022improved} further explore this model. Among them, the most closely related is~\cite{DBLP:conf/esa/DameriusKLX023}, which studies a single-machine scheduling problem where each job has an upper and lower processing time, and testing a job reveals its lower bound. This mirrors the upgrade operation in our setting. The goal is to select at most $k$ jobs to test in order to minimize the total completion time. They present an FPTAS for this problem. However, due to the different nature of constraints, the techniques used in our work diverge substantially from theirs.

Other related work focuses on optimization in network upgrades. For instance,~\cite{DBLP:journals/jgaa/BaderB22} investigates how to efficiently find alternative edges when the cost of a network link changes, and~\cite{luo2024efficient} studies dynamic algorithms for maintaining a minimum spanning tree as the graph evolves. Notably, these works are concerned with adapting to changes in edge costs or network structure after upgrades have occurred. In contrast, our model focuses on deciding which edges to upgrade in order to optimize a global objective.

\subsection{Paper Organization}\label{sec:org}

To build intuition for our algorithm and its analysis, we will focus on a special case of our model in the main body, where each agent's matroid is a \emph{graphic matroid}. We refer to this special case as the \emph{multiagent graphic upgrading} problem (MGUP). In this setting, the ground elements correspond to edges in a connected graph, and each matroid basis forms a spanning tree. This case not only arises naturally in applications such as network upgrading but also provides a concrete and visual framework to illustrate the key ideas behind our algorithm and analysis. 


\cref{sec:pre} introduces the formal definition of MGUP. In~\cref{sec:alg}, we present an optimal greedy algorithm for MGUP together with the basic analysis framework. The subsequent two sections (\cref{sec:single,sec:multi}) provide more detailed analysis. In~\cref{sec:graph_exten}, we demonstrate that our results can be naturally extended to other models, such as those with minimax objectives or interval fairness constraints. The paper concludes in~\cref{sec:con}.Additionally, the general matroid results, their corresponding extensions, and the connection to budget-constrained MST are deferred to~\cref{sec:matroid},~\cref{sec:exten}, and~\cref{sec:connection}, respectively.



\section{Preliminaries}\label{sec:pre}





As mentioned in~\cref{sec:org}, the main body of this paper focuses on the \emph{multiagent graphic upgrading problem} (MGUP), a special case of MMUP. In this section, we formalize the model. We begin by outlining the problem setting, specifying the input, output, and objective, with a simple example provided in~\cref{fig:example}. We then introduce the notation and assumptions that will be used throughout the paper.

\begin{itemize}[left=10pt]
    \item \textbf{Input:}
    \begin{itemize}
        \item $n$ disjoint connected graphs $\{G^{(i)} = (V^{(i)}, E^{(i)})\}_{i\in [n]}$, each representing an agent.
        \item Let $E := \bigcup_{i \in [n]} E^{(i)}$ denote the set of all edges across the $n$ graphs. Each edge $e\in E$ has two non-negative costs: a default cost $\dc(e)$ and an upgraded cost $\uc(e)$, with $\dc(e) \geq \uc(e)$.
        \item An upgrade quota $k$ specifying the maximum number of edges that can be upgraded.
        \item A non-decreasing convex function $F_i$ for each agent $i \in [n]$.
    \end{itemize}

    \item \textbf{Output:}
    \begin{itemize}
        \item An edge subset $S \subseteq E$ with $|S| \leq k$ to be upgraded.
    \end{itemize}

    \item \textbf{Objective:}
    \begin{itemize}
        \item Minimize the total cost:
        \(
        \obj(S) = \sum_{i \in [n]} F_i\left(\delta_S(G^{(i)})\right),
        \)
        where $\delta_S(G^{(i)})$ is the cost of the minimum spanning tree (MST) of $G^{(i)}$ (i.e., the minimum-cost basis of the corresponding graphic matroid) under the modified cost function: $c_S(e):=\uc(e)$ if $e\in S$; otherwise, $c_S(e)=\dc(e)$.
    \end{itemize}
\end{itemize}



Let $\Ti_S$ denote the minimum spanning tree (MST) of graph $G^{(i)}$ under the edge cost function $c_S$, and define
\[
\delta_S(G^{(i)}) := \sum_{e \in \Ti_S} c_S(e)
\]
as its total cost. Furthermore, let \[\cT_S := \bigcup_{i \in [n]} \Ti_S\] denote the union of all agents' MSTs.

\paragraph{Uniqueness Assumption.}
Without loss of generality, we assume that $\Ti_{S}$ is unique for every $S$, which can be ensured by applying a fixed tie-breaking rule (e.g., a consistent ordering over the edges). This uniqueness assumption simplifies the subsequent analysis.

\begin{figure*}[htb]
    \centering
    \includegraphics[width=0.7\linewidth]{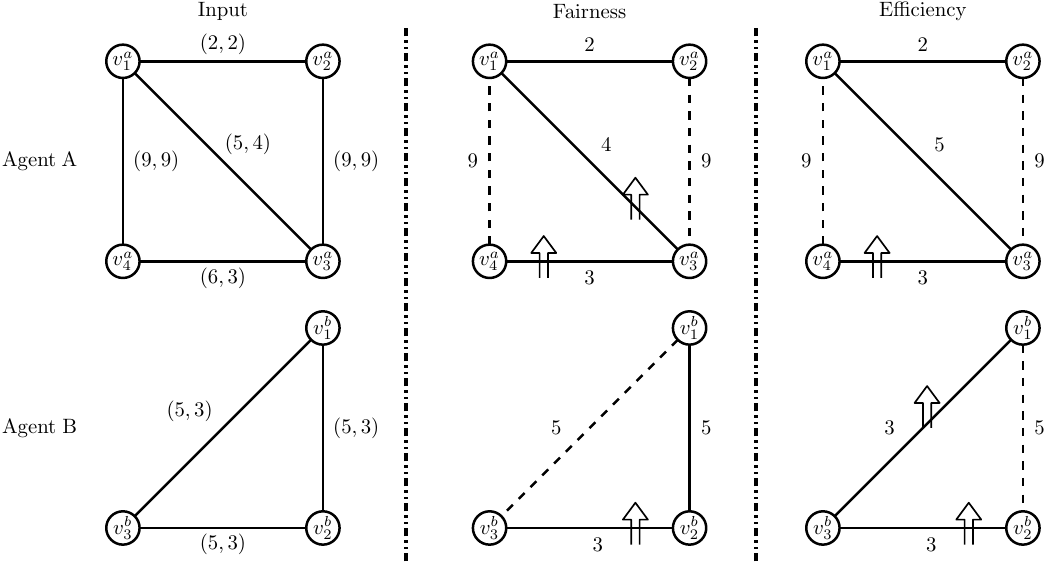}
    \caption{An illustration of MGUP. The original input graphs are shown on the left: two agents, agent A and agent B, each own a disjoint undirected graph. Each edge is labeled with a pair of costs $(\cdot,\cdot)$, where the first is the default cost and the second is the upgraded cost. The upgrade quota is $3$. In the center of the figure, we set each $F_i(x) = 2^x$, aiming to optimize fairness via a minimax-style objective. The optimal upgrade solution under this fairness objective is shown: upgraded edges are marked with upright arrows, and their costs are reduced to the upgraded values. Solid lines represent edges included in the MSTs after upgrading; dashed lines are not part of the MSTs. To ensure fairness, the optimal solution allocates two upgrades to agent A and one to agent B. After upgrading, agent A’s MST has a cost of $9$ and agent B’s MST has a cost of $8$, giving a total objective value of $2^9 + 2^8 = 768$. The right part of the figure shows the optimal solution under the efficiency objective, where $F_i(x) = x$. Here, the solution allocates one upgrade to agent A and two to agent B. After upgrading, the MST costs are $10$ for agent A and $6$ for agent B, yielding a total objective value of $10 + 6 = 16$.
 }
    \label{fig:example}
\end{figure*}

We will leverage the following property of minimum spanning trees in the analysis. 

\begin{lemma}[Cut Property~\cite{DBLP:books/daglib/0023376}]\label{lem:mst-cut}
   Given a connected graph $G(V,E)$ and its minimum spanning tree $\T$ corresponding to an edge cost function $c(\cdot)$, we have for any vertex subset $X\subseteq V$, there exists at least one edge $e\in \T \cap \cC(X,V\setminus X)$ such that $c_e = \min_{e'\in \cC(X,V\setminus X)} c_{e'}$, where $\cC(X,V\setminus X):=\{(x,y)\in E \mid x\in X, y\notin X\}$ denotes the set of edges crossing the cut $(X,V\setminus X)$.   
\end{lemma}

\section{An Optimal Greedy Algorithm}\label{sec:alg}



This section presents a greedy algorithm for solving MGUP. As outlined in~\cref{alg:mst}, we begin by computing an MST for each follower under the upgraded cost function~$\cc$. Due to our uniqueness assumption, each $\Ti_E$ is unique. Then, in each iteration, we greedily select an edge from the union of these MSTs to upgrade, choosing the one that minimizes the overall objective.


\begin{algorithm}[tb]
\caption{Greedy Upgrading for MGUP}
\label{alg:mst}
\begin{algorithmic}[1] 
\REQUIRE Connected graphs $\{G^{(i)}=(V^{(i)},E^{(i)})\}_{i\in [n]}$, two edge cost functions $\hc(\cdot), \cc(\cdot)$, upgrade quota $k$, and functions $\{F_i\}_{i\in [n]}$. \\
\ENSURE An edge subset $S$. 
\STATE For each graph $G^{(i)}$, compute $\Ti_E$ (i.e., the unique MST under the upgraded cost $\cc$). 
\STATE Let $\cT_E := \bigcup_{i \in [n]} \Ti_{E} $; Initialize $S\gets \emptyset$.
\WHILE{ $|S| < k$ }
\STATE Pick an edge $e\in \cT_E$ such that $\obj(S\cup\{e\})$ is minimized; $S\gets S\cup \{e\}$.
\ENDWHILE
\STATE \textbf{return} $S$.
\end{algorithmic}
\end{algorithm}

\begin{theorem}\label{thm:mst}
    Given any MGUP instance, \cref{alg:mst} returns an optimal subset $S$ in polynomial time. Moreover, with an efficient implementation, the algorithm runs in $O(k|V|^2)$ time, where $V := \bigcup_{i \in [n]} V^{(i)}$ denotes the union of all vertex sets across the input graphs.
\end{theorem}

Since the algorithm initially computes the candidate edge set $\cT_E$ and only selects upgrade edges from this set, proving the correctness of \cref{thm:mst} requires showing that there always exists an optimal solution $S^*$ that is a subset of $\cT_E$. In fact, we can establish an even stronger result in the following. 

Let $\Ti_{\emptyset}$ denote the MST of agent~$i$ under the default cost function $\dc$ (i.e., when no edges are upgraded, corresponding to $S = \emptyset$), and let $\cT_{\emptyset} := \bigcup_{i \in [n]} \Ti_{\emptyset}$ denote the union of all such MSTs. We show that there exists an optimal solution whose MST union includes no edges outside of $\cT_{\emptyset} \cup \cT_E$. 



\begin{lemma}[Inclusion Property]\label{lem:inclusion}
    An optimal upgrade set $S^* \subseteq \cT_E$ always exists such that, for any agent $i \in [n]$, the MST $\Ti_{S^*}$ contains only edges from $\Ti_{\emptyset} \cup \Ti_E$, i.e., 
    \(
    \Ti_{S^*} \subseteq \Ti_{\emptyset} \cup \Ti_E.
    \)
\end{lemma}
\begin{proof}
    Consider an optimal edge subset $S^*$ and the corresponding MST $\Ti_{S^*}$ for some agent $i$.  
Suppose there exists an edge $e \in \Ti_{S^*} \setminus \left(\Ti_\emptyset \cup \Ti_E\right)$.  
We will show that $e$ can always be replaced by an edge from $\Ti_\emptyset \cup \Ti_E$ without increasing the objective value or the number of upgraded edges.

Removing $e$ from $\Ti_{S^*}$ partitions the spanning tree into two connected components; let $X$ and $Y$ denote the vertex sets of these components.  
By~\cref{lem:mst-cut}, $\Ti_\emptyset$ contains an edge $\he$ with the minimum default cost $\hc$ in the cut set $\cC(X, Y)$, and $\Ti_E$ contains an edge $\ce$ with the minimum upgraded cost $\cc$ in the same cut set.  
Therefore, regardless of whether $e$ is an upgraded edge or not, it can be replaced by either $\he$ or $\ce$ without increasing the number of upgraded edges or the objective value.  
This shows that we can always assume the optimal solution only upgrades edges from $\cT_E$, and the resulting MSTs are subgraphs of $\cT_{\emptyset} \cup \cT_E$.  
\end{proof}

The lemma above justifies why~\cref{alg:mst} considers upgrading only the edges in the candidate set \(\cT_E\), and we can, without loss of generality, assume that the graph contains no edges other than those in \(\cT_{\emptyset} \cup \cT_E\).
Moreover, the inclusion property offers a more structured view of optimal solutions: the union of MSTs corresponding to an optimal solution, $\cT_{S^*}$, can be interpreted as the result of removing a set of $k$ edges $R^*$ from the default-cost union $\cT_{\emptyset}$, and then adding $k$ edges $S^* \subseteq \cT_E$ from the upgraded-cost union\footnote{Without loss of generality, in this paper, we assume $k$ is at most the number of edges in $\cT_E$; otherwise, the problem becomes trivial, as upgrading all edges in $\cT_E$ would yield an optimal solution.}.  
In other words, we can view the construction of $\cT_{S^*}$ as starting from $\cT_{\emptyset}$ and replacing the edges in $R^*$ with those in $S^*$.  
Therefore, an optimal solution can be naturally represented as a pair of swap edge sets $(R^*, S^*)$, where $S^* \subseteq \cT_E$ replaces $R^* \subseteq \cT_{\emptyset}$ to form the final MST union $\cT_{S^*}$.

\begin{definition}[$t$-Optimal Solution]\label{def:toptimal}
    Given an MGUP instance with an upgrade quota $k$, for any $1 \leq t \leq k$, we define a \emph{$t$-optimal solution} as the optimal solution to the instance when the upgrade quota is reduced to $t$. Let $(R^*_t, S^*_t)$ denote the corresponding edge set pair in this case.
\end{definition}



\begin{definition}[Nested Optimal Sequence]\label{def:nest}
    Given an MGUP instance, a sequence of solutions \(\langle (R^*_1, S^*_1), (R^*_2, S^*_2), \ldots, (R^*_t, S^*_t) \rangle\) is called a \emph{nested optimal sequence} if it satisfies
    \[
    R^*_1 \subseteq R^*_2 \subseteq \cdots \subseteq R^*_t
    \quad \text{and} \quad
    S^*_1 \subseteq S^*_2 \subseteq \cdots \subseteq S^*_t~.
    \]
\end{definition}

\begin{lemma}[Nestedness]\label{lem:nest}
    Given an MGUP instance, for any nested optimal sequence 
    \(\langle(R^*_1, S^*_1), (R^*_2, S^*_2), \ldots, (R^*_t, S^*_t)\rangle\), 
    there always exists a \((t+1)\)-optimal solution \((R^*_{t+1}, S^*_{t+1})\) that preserves the nestedness property, i.e., 
    \(
    R^*_t \subseteq R^*_{t+1}\) and \( S^*_t \subseteq S^*_{t+1}~.
    \)
\end{lemma}

\cref{lem:nest} plays a crucial role in proving the optimality of the greedy algorithm and is the most technically challenging part of our analysis. We will defer the discussion and proof of this lemma to the later sections. For now, we assume the correctness of the nestedness lemma and proceed to present the proof of~\cref{thm:mst}.

\begin{proofof}{\cref{thm:mst}}
We prove the theorem by induction.  
For the base case when $t = 1$, by~\cref{lem:inclusion}, a $1$-optimal solution can be obtained by selecting the edge $e \in \cT_E$ that results in the greatest decrease in the objective value.  
Therefore,~\cref{alg:mst} correctly returns $(R^*_1, S^*_1)$ in the first iteration.

Now assume as the induction hypothesis that after the first $t$ iterations,~\cref{alg:mst} maintains a nested optimal sequence 
\(
\langle (R^*_1, S^*_1), (R^*_2, S^*_2), \ldots, (R^*_t, S^*_t) \rangle~.
\)
Then by~\cref{lem:nest}, there exists a $(t+1)$-optimal solution $(R^*_{t+1}, S^*_{t+1})$ such that 
\(
R^*_t \subseteq R^*_{t+1} \text{ and } S^*_t \subseteq S^*_{t+1}~.
\)
Hence, in iteration $t+1$, greedily selecting the edge $e \in \cT_E$ that maximizes the decrease in the objective value produces a valid $(t+1)$-optimal solution.

    For the running time, we observe that, due to the inclusion lemma, the algorithm only considers edges in $\cT_\emptyset \cup \cT_E$ throughout the process of finding the optimal solution. Furthermore, by the nestedness lemma, the algorithm essentially starts from $\cT_\emptyset$ and iteratively replaces an edge in the current solution with an edge selected from $\cT_E$ in each round. Using Fibonacci heaps to compute $\cT_\emptyset$ and $\cT_E$ takes $O(|E| + |V| \log |V|)$ time, and each greedy edge selection and replacement from $\cT_E$ takes $O(|V|^2)$ time. Therefore, the total running time is $O(k |V|^2)$.
\end{proofof}
\section{Single Agent Upgrading}\label{sec:single}


This section focuses on the proof of~\cref{lem:nest}.
We start by considering the single-agent case where the number of agents $n=1$. Note that when $n=1$, the objective reduces to minimizing $\delta_S(G^{(1)})$ as function $F_1$ is non-decreasing. For simplicity, in this section, we will drop the superscript and let $\T_S$ denote the minimum spanning tree of the single agent.

\begin{lemma}\label{lem:single}
The nestedness property (\cref{lem:nest}) holds in the single-agent case.
\end{lemma}

We observe that the key to proving the nestedness is showing that an arbitrary $(R_1^*, S_1^*)$ is always contained in some optimal solution $(R_k^*, S_k^*)$ for any $k$ (\cref{lem:single_1}). Once this claim is established,~\cref{lem:single} can be easily proven via an inductive argument.

\begin{lemma}\label{lem:single_1}
Given a single-agent MGUP instance, for any $(R_1^*, S_1^*)$ and any $k\geq 2$, there always exists a $k$-optimal solution $(R_k^*, S_k^*)$ such that $R_1^* \subseteq R_k^*$ and $S_1^* \subseteq S_k^*$.  
\end{lemma}

Consider a $1$-optimal solution $(R_1^*=\{r\}, S_1^*=\{s\})$ and a $k$-optimal solution $(R_k^*, S_k^*)$. The basic idea of proving~\cref{lem:single_1} is to show that if  $r\notin R_k^* $ or $s\notin S_k^*$, we can always use $(r,s)$ to replace some edges in $(R_k^*, S_k^*)$ without increasing the objective value. 

As mentioned earlier, we can w.l.o.g. assume that $k$ is at most the size of a spanning tree and all upgraded edges are in the resulting minimum spanning tree, i.e., $S_k^* \subseteq \T_{S_k^*}$.
Let $\{X,Y\}$ be the vertex sets of the two connected components formed by removing edge $s$ from $\T_{S_1^*}$, and $\{Z_1,\ldots,Z_{k+1}\}$ be the $k+1$ connected components' vertex sets formed by removing edge set $S_k^*$ from $\T_{S_k}$. Say set $X$ separates set $Z_i$ if $\emptyset \subset X\cap Z_i \subset Z_i$. We have the following observation.


\begin{observation}\label{obs:separate}
    There exists at most one $Z_i\in \{Z_1,\ldots,Z_{k+1}\}$ which is separated by $X$.
\end{observation}
\begin{proof}
        Since $\T_{\emptyset}$ is unique by the uniqueness assumption, the inclusion property (\cref{lem:inclusion}) implies that both $\T_{S_1^*}$ and $\T_{S_k^*}$, after removing their upgraded edges, must be subgraphs of $\T_{\emptyset}$. Moreover, by the definition of swap edge set pairs, for $t = 1$ or $k$, we have
\[
\T_{S_t^*} \setminus S_t^* = \T_{\emptyset} \setminus R_t^*.
\]

Therefore, removing $R_1^*$ from $\T_{\emptyset}$ separates the vertex sets $X$ and $Y$ (i.e., any two vertices in $X$ and $Y$ are disconnected in $\T_{\emptyset} \setminus R_1^*$), while removing $R_k^*$ partitions the graph into $(k+1)$ connected components $\{Z_1, \ldots, Z_{k+1}\}$.

Assume for contradiction that two connected components, say $Z_1$ and $Z_2$, are separated by the cut $(X,Y)$. Then, removing the combined edge set $R_1^* \cup R_k^*$ (with size at most $k + 1$) from $\T_{\emptyset}$ would produce at least $k + 3$ connected components:
\[
\{Z_1 \cap X, Z_1 \cap Y, Z_2 \cap X, Z_2 \cap Y, Z_3, Z_4, \ldots, Z_{k+1}\}~;
\]
otherwise, $X$ and $Y$ would not be fully disconnected after removing $R_1^*$, contradicting the claim above.
This leads to a contradiction, as removing at most $k + 1$ edges from a spanning tree can yield at most $k + 2$ connected components. Hence, the observation is proved.
\end{proof}
\begin{proofof}{\cref{lem:single_1}}
     Due to~\cref{obs:separate}, we can distinguish two cases (see illustrations in~\cref{fig:single_1}): 
     \begin{enumerate}
         \item[(1)] No set in $\{Z_1,\ldots,Z_{k+1}\}$ is separated by $X$;
         \item[(2)] Exactly one $Z_i$ is separated by $X$.
     \end{enumerate}

    \begin{figure}[tb]
        \centering
        \includegraphics[width=0.9\linewidth]{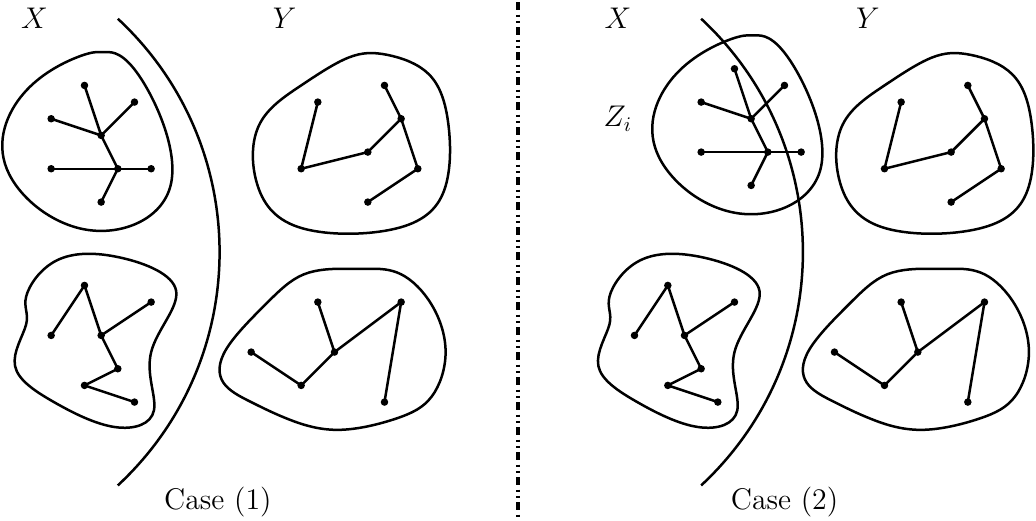}
        \caption{Illustration for the two cases of \cref{lem:single_1}.}
        \label{fig:single_1}
    \end{figure}

    For case (1), since no set is separated by $X$, we see that removing $S_k^*$ from $\T_{S_k^*}$ disconnects $X$ and $Y$, i.e., \[\cC(X,Y) \cap \left( \T_{S_k^*} \setminus S_k^* \right) = \emptyset~.\] According to the definition of $(X,Y)$, we have $r\in \cC(X,Y)$, implying that edge $r$ has been swapped out in $\T_{S_k^*}$, i.e., $r \in R_k^*$. As $(R_1^*,S_1^*)$ is a $1$-optimal solution, by~\cref{lem:mst-cut}, edge $s$ has the minimum $\cc$ cost among $\cC(X,Y)$. Thus, if $s \notin \T_{S_k^*}$, we can always find an upgraded edge in $\T_{S_k^*} \cap \cC(X,Y)$ that can be replaced by $s$ without violating the feasibility or increasing the cost.

    For case (2), we have $ r \in \T_{S_k^*} $ because 
    \[\cC(X,Y) \cap \left(\T_{S_k^*} \setminus S_k^* \right) = \cC(X,Y) \cap \left(\T_{S_k^*} \cap \T_{\emptyset} \right) \neq \emptyset\]    
    and the definition of $(X,Y)$ implies that \[\cC(X,Y) \cap \T_{\emptyset} = \{r\}~.\] 
    Let $L$ denote the cycle formed by adding edge $s$ to the tree $\T_{S_k^*}$. Note that $s$ may already be in $\T_{S_k^*}$, where it forms a self-cycle.

    We further distinguish two subcases based on whether this cycle $L$ contains edge $r$ or not. In the subcase where $r \in L$, we swap $r$ with $s$, i.e., \[R_k'\leftarrow R_k^* \cup \{r\}\] and \[ S_k'\leftarrow S_k^* \cup \{s\}~. \] Notice that this operation obtains a feasible spanning tree but the number of upgraded edges increases by 1. To complete the proof for this subcase, we demonstrate the existence of another pair $(r',s')\in (R_k',S_k') $ such that swapping back $s'$ with $r'$ giving a feasible spanning tree and the final cost is non-increasing.

    Let $\{Z_1',\ldots,Z_{k+1}'\}$ be the formed components when removing $S_k^*$ from $\T_{S_k'}$ (note that the newly added edge $s$ is not removed here).  
    Find a leaf component $Z_j'$, which means that \[|\cC(Z_j', V\setminus Z_j' ) \cap S_k^* |=1~.\] Since $\T_{S_k'}$ is a spanning tree, we know that at least two such leaf components exist. Consequently, we can assume without loss of generality that edge $s$ does not belong to $Z_j'$ and that $Z_j'$ remains a connected component even if $s$ is removed.
    
    
    As $Z_j'$ is disconnected from $V\setminus Z_j'$ in $\T_{S_k'} \cap \T_{\emptyset}$, all edges in $\cC(Z_j', V\setminus Z_j' ) \cap \T_{\emptyset}$ have been swapped out, i.e., \[\cC(Z_j', V\setminus Z_j' ) \cap \T_{\emptyset} \subseteq R_k'~.\] Let edge $s'$ be the unique edge in $ \cC(Z_j', V\setminus Z_j' ) \cap S_k^* $ (as $Z_j'$ is a leaf component). Clearly, adding $s'$ into tree $ \T_\emptyset$ forms a cycle $L'$ that contains at least one edge $r'\in \cC(Z_j', V\setminus Z_j' ) \cap \T_{\emptyset}$. Since $(r, s)$ is $1$-optimal, we have
    \(\cc(s) - \hc(r) \leq \cc(s') - \hc(r').  \)
    As $Z_j'$ is a leaf component, replacing $s' \in \T_{S_k'}$ with $r'$ always gives a feasible spanning tree. By letting \[R_k''\gets R_k' \setminus \{r'\}\] and \[S_k''\gets S_k' \setminus \{s'\}~,\] we obtain a feasible solution that upgrades at most $k$ edges and admits a cost \[\T_{S_k^*} + (\cc(s) - \hc(r)) - (\cc(s') - \hc(r')) \leq  \T_{S_k^*}~.\]

    In the second subcase where $r \notin L$, the cycle $L$ must contain another upgraded edge $s' \in \cC(X,Y) \cap S_k^*$. Indeed, if every edge on $L$ were non-upgraded (i.e., belonged to $\T_{\emptyset} \setminus \{r\}$), then $T_{S_1^*} = \T_{\emptyset} \setminus \{r\} \cup \{s\}$ would contain a cycle, a contradiction.

    By~\cref{lem:mst-cut}, edge $s$ has the minimum upgraded cost among $\cC(X,Y)$. We can replace $s'$ with $s$ without violating the feasibility or increasing the overall cost. The last piece is to show that we can replace $r$ with some other edge $ r' \in R_k^*$ that satisfies $\hc(r') \leq \hc(r)$.

    Consider the cycle formed by adding edge $s$ to $\T_{\emptyset}$. Clearly, both $r$ and $s$ belong to this cycle. Moreover, the cycle must contain at least one edge $r' \in R_k^*$ (i.e., an edge that was swapped out in the current solution), since otherwise the cycle would already exist in the current tree formed by $(R_k'', S_k'')$. We observe that $r$ has the maximum default cost among all edges in the cycle; otherwise, the pair $(r, s)$ would not be 1-optimal, and some other edge in $\T_{\emptyset}$ would have been swapped out instead. Therefore, we must have $\hc(r') \leq \hc(r)$, which implies that $r$ can be replaced with some edge $r' \in R_k^*$ without increasing the objective value.  
\end{proofof}


\begin{proofof}{\cref{lem:single}}
We establish the lemma by induction. The base case is verified by~\cref{lem:single_1}, which proves the statement for $t=1$. Now assume that~\cref{lem:single} holds for some $t \geq 1$ on any MGUP instance. We aim to prove the inductive step for the case $t+1$.

Consider an MGUP instance with a nested optimal sequence of length $t+1$, denoted as 
\[
\sigma = \langle(R^*_1, S^*_1), (R^*_2, S^*_2), \ldots, (R^*_{t+1}, S^*_{t+1})\rangle.
\]
This sequence can be viewed as a step-by-step upgrade order: first the edge in $S^*_1$, then the edge in $S^*_2 \setminus S^*_1$, and so on. 
We treat the graph obtained after upgrading the edge in $S^*_1$ as a new MGUP instance, where the default cost of each edge $e \in S^*_1$ is reduced to its upgraded cost. Formally, define the new default cost function $\hc'$ as follows: 
\[
\hc'(e) = 
\begin{cases}
\cc(e) & \text{if } e \in S^*_1, \\
\hc(e) & \text{otherwise}.
\end{cases}
\]
That is, in the new instance, edges in $S^*_1$ no longer require an upgrade.

Now define a new nested sequence of length $t$ by removing $(R^*_1, S^*_1)$ from each subsequent solution in $\sigma$:
\[
\sigma' = \langle(R^*_2 \setminus R^*_1, S^*_2 \setminus S^*_1), \ldots, (R^*_{t+1} \setminus R^*_1, S^*_{t+1} \setminus S^*_1)\rangle.
\]
Clearly, $\sigma'$ forms a nested optimal sequence in the newly constructed instance.
By the induction hypothesis, there exists a $(t+1)$-optimal solution $(R, S)$ in this new instance such that 
\(
R^*_{t+1} \setminus R^*_1 \subseteq R \text{ and } S^*_{t+1} \setminus S^*_1 \subseteq S.
\)
Mapping this solution back to the original instance, we obtain a nested $(t+2)$-optimal solution $(R \cup R^*_1, S \cup S^*_1)$, completing the inductive proof.
\end{proofof}
\section{Multiple Agent Upgrading}\label{sec:multi}


In this section, we show that the nestedness lemma holds in the general case where the number of agents \( n > 1 \). We begin by analyzing a special case in which each \( F_i \) is an \emph{identity} function, i.e., \( F_i(x) = x \). Building on this result, we then extend the proof to the case where the \( F_i \) are general non-decreasing convex functions.


\subsection{Identity Functions}\label{sec:multi-l1}

This subsection considers the case where each function \( F_i \) is the identity, and the objective is to minimize the total MST cost across all agents, i.e., \( \sum_{i \in [n]} \delta_S(G^{(i)}) \).

\begin{lemma}\label{lem:multi-l1}
    The nestedness property (\cref{lem:nest}) holds when each function \( F_i \) is the identity.
\end{lemma}

\begin{proof}
    Building on the results from the previous section, this lemma can be readily proven. Since the objective is the sum cost of all agents, we can introduce $n-1$ zero-cost edges to connect the given $n$ connected graphs, thereby reducing the problem to a single-agent upgrade instance. 

    More formally, we add an arbitrary edge $e_i$ with $\hc(e_i)=\cc(e_i)=0$ to connect graph $G^{(i)}$ and $G^{(i+1)}$ for each $i<n$. Observe that minimizing $\sum_{i \in [n]} \delta_S(G^{(i)})$ of the original instance is equivalent to minimizing the MST cost for a single agent in the newly constructed large graph. Notably, the optimal solution, as well as the behavior of our algorithm, remains unchanged. Therefore, by~\cref{lem:single}, the nestedness lemma holds. 
\end{proof}

\subsection{General Functions}\label{sec:multi-lq}

This subsection considers the general case where each \( F_i \) is an arbitrary non-decreasing convex function. Since the overall objective is the sum of these functions, we observe that any optimal solution \( (R^*_k, S^*_k) \) remains optimal when restricted to each individual agent's graph. 
Formally, let \( (R^*_k(i), S^*_k(i)) \) denote the sets of edges in \( (R^*_k, S^*_k) \) that belong to agent \( i \)'s graph \( G^{(i)} \). We have the following observation:

\begin{observation}\label{obs:agents}
For any \( k \)-optimal solution \( (R^*_k, S^*_k) \), if we consider agent \( i \) as a single-agent instance, then \( (R^*_k(i), S^*_k(i)) \) is a \( t \)-optimal solution for that instance, where \( t = |S^*_k(i)| \).
\end{observation}
\begin{proof}
    This observation holds because each \( F_i \) is non-decreasing. If \( (R^*_k(i), S^*_k(i)) \) were not optimal for agent \( i \), then replacing it with the \( t \)-optimal solution on \( G^{(i)} \) would strictly reduce the overall objective value, contradicting the optimality of \( (R^*_k, S^*_k) \). 
\end{proof}

Let \(\Delta^{(i)}(t)\) denote the minimum MST cost of graph \(G^{(i)}\) when at most \(t\) edges are allowed to be upgraded, i.e.,
\(
\Delta^{(i)}(t) = \min_{|S| \leq t} \delta_S(G^{(i)}).
\)
Another key observation is that \(\Delta^{(i)}(t)\) is convex.

\begin{observation}
For any agent \(i\), the function \(\Delta^{(i)}\) is non-increasing, and for any integer \(t > 1\), it satisfies
\[
\Delta^{(i)}(t-1) - \Delta^{(i)}(t) \geq \Delta^{(i)}(t) - \Delta^{(i)}(t+1)~.
\]
\label{obs:convex}
\end{observation}
\begin{proof}
    The non-increasing property is straightforward, as allowing more edges to be upgraded can only decrease or maintain the minimum MST cost. 

For convexity, assume for contradiction that there exists a smallest counterexample \( t \) such that
\(
\Delta^{(i)}(t-1) - \Delta^{(i)}(t) < \Delta^{(i)}(t) - \Delta^{(i)}(t+1),
\)
where the minimality of \( t \) implies that convexity holds for all smaller values. 
We will construct an MGUP instance with identity functions to show that if such a \( t \) exists, it would contradict~\cref{lem:multi-l1}.

The instance is constructed as follows. There are two agents: the first agent has the graph \( G^{(i)} \) with the same default and upgraded costs as in the original instance.
The second agent has a single-edge graph with  
\[\hc(e) = \frac{\Delta^{(i)}(t-1) - \Delta^{(i)}(t+1)}{2} \] and \(\cc(e) = 0\).
Each cost function \( F_i \) is an identity function.

Consider the \( t \)-optimal solution for this instance. By construction, the objective is minimized by upgrading \( t-1 \) edges for the first agent and the only edge for the second agent, yielding an objective value of \( \Delta^{(i)}(t-1) \). This \( t \)-optimal solution is unique: if we instead upgrade \( t \) edges in the first agent’s graph, the objective becomes
\[
\Delta^{(i)}(t) + \frac{\Delta^{(i)}(t-1) - \Delta^{(i)}(t+1)}{2}~,
\]
which is strictly greater than \( \Delta^{(i)}(t-1) \) under our contradiction assumption.

Now consider the \((t+1)\)-optimal solution. It would upgrade \( t+1 \) edges in the first agent’s graph, resulting in an objective of
\[
\Delta^{(i)}(t+1) + \frac{\Delta^{(i)}(t-1) - \Delta^{(i)}(t+1)}{2},
\]
which, again by the contradiction assumption, is strictly less than
\(
\Delta^{(i)}(t),
\)
the objective obtained by upgrading \( t \) edges for the first agent and one edge for the second agent.

This contradicts the nestedness property established in~\cref{lem:multi-l1}, which implies that a \((t+1)\)-optimal solution should extend the \( t \)-optimal solution. Therefore, no such \( t \) violating convexity can exist.
\end{proof}

\begin{lemma}\label{lem:multi-lq}
    The nestedness property (\cref{lem:nest}) holds when each function $F_i$ is non-decreasing and convex.
\end{lemma}

\begin{proof}

Given an MGUP instance and a nested optimal sequence 
\(\langle(R_1, S_1), (R_2, S_2), \ldots, (R_t, S_t)\rangle\), we prove the existence of a nested \((t+1)\)-optimal solution by showing that there exists a \((t+1)\)-optimal solution \((R^*, S^*)\) such that \( |S^*(i)| \geq |S_t(i)| \) for every agent \( i \), where \( S^*(i) \) and \( S_t(i) \) denote the sets of upgraded edges in agent \( i \)'s graph \( G^{(i)} \) for solutions \((R^*, S^*)\) and \((R_t, S_t)\), respectively. Then, by applying \cref{obs:agents} and the nestedness lemma for the single-agent case (\cref{lem:single}), we can conclude that the lemma holds.


As the optimal sequence \(\langle(R_1, S_1), (R_2, S_2), \ldots, (R_t, S_t)\rangle\) is nested, each step in the sequence adds exactly one new upgraded edge. To maintain optimality, the newly added edge must be the one that yields the maximum decrease in the objective value. In other words, a nested optimal sequence implicitly reflects greedy edge selection at each step. We will leverage this greedy property to prove the claim above.

Suppose that there exists an agent \( i \) such that \( |S^*(i)| < |S_t(i)| \). For notational convenience, let \( u_i^* = |S^*(i)| \) and \( u_i = |S_t(i)| \). Since the solution \( S^* \) upgrades \( t+1 \) edges while \( S_t \) upgrades \( t \) edges, there must exist another agent \( j \) with \( u_j^* > u_j \).
We next show that if we construct a new solution \( S' \) by upgrading one fewer edge for agent \( j \) and one more edge for agent \( i \) (i.e., \( |S'(i)| = u_i^* + 1 \) and \( |S'(j)| = u_j^* - 1 \)), then the objective value does not increase.

Since the nested optimal sequence exhibits greedy behavior and \( u_i^* < u_i \), the edge selected by the greedy algorithm as the \((u_i^* + 1)\)-th upgrade for agent \( i \) must have yielded the maximum decrease in the objective among all candidates at that step. Let \( \bar{u}_j \) denote the number of edges agent \( j \) had upgraded at that point. Note that we have $\bar{u}_j \leq u_j$ and thus $ \bar{u}_j \leq u_j^* $.  By the greedy property, we have:
\[
F_i(\Delta^{(i)}(u_i^*)) - F_i(\Delta^{(i)}(u_i^* + 1)) 
\geq 
F_j(\Delta^{(j)}(\bar{u}_j)) - F_j(\Delta^{(j)}(\bar{u}_j + 1))~.
\]
Using the convexity of \( \Delta^{(j)} \) and the fact that \( F_j \) is non-decreasing and convex, we obtain:
\[
F_j(\Delta^{(j)}(\bar{u}_j)) - F_j(\Delta^{(j)}(\bar{u}_j + 1))
\geq 
F_j(\Delta^{(j)}(u_j^*)) - F_j(\Delta^{(j)}(u_j^* + 1))~.
\]
Combining the two inequalities yields:
\[
F_i(\Delta^{(i)}(u_i^*)) - F_i(\Delta^{(i)}(u_i^* + 1))
\geq 
F_j(\Delta^{(j)}(u_j^*)) - F_j(\Delta^{(j)}(u_j^* + 1))~.
\]
This implies that swapping one unit of upgraded quota from agent \( j \) to agent \( i \) will not increase the objective value, as the total decrease in the objective function is non-negative: 
\begin{align*}
\;\;\; F_i&(\Delta^{(i)}(u_i^*)) - F_i(\Delta^{(i)}(u_i^* + 1)) \\
&\;-\;\Big(F_j(\Delta^{(j)}(u_j^*)) - F_j(\Delta^{(j)}(u_j^* + 1))\Big) \geq 0~.
\end{align*}


By repeatedly applying this argument, we can always transform any \((t+1)\)-optimal solution into one that satisfies \( |S^*(i)| \geq |S_t(i)| \) for every agent \( i \). This completes the proof of the lemma.
\end{proof}

Combining~\cref{lem:multi-l1} and~\cref{lem:multi-lq} proves the nestedness lemma (\cref{lem:nest}).

\section{Further Extensions}\label{sec:graph_exten}

In this section, we discuss two further extensions of MGUP and show that greedy still remains optimal in these extended models.

\subsection{Minimax Objective}

We consider a variant of the model where the objective changes from the sum of the \( F_i \)'s to their maximum. That is, the problem becomes:
\begin{equation}
\begin{aligned}
    \min_{S \subseteq E,\, |S| \le k} &\max_{i \in [n]} F_i\left( \delta_S(G^{(i)}) \right) \\
    \text{s.t.} \quad \delta_S(G^{(i)}) &= \sum_{e \in \Ti_S} c_S(e) \quad \forall i \in [n]~.
\end{aligned}
\tag{Minimax-MGUP}
\end{equation}

\begin{theorem}\label{thm:graph_minmax}
    Given any Minimax-MGUP instance,~\cref{alg:mst} returns an optimal subset in polynomial time.
\end{theorem}

\begin{proof}
    The proof is nearly identical to that of~\cref{thm:mst}; the only difference lies in the final part of the nestedness proof for multiple agents (corresponding to~\cref{lem:multi-lq}). All other lemmas, including those for the single-agent case, remain applicable. The difference here arises because the objective is now a maximum function rather than a sum, and is therefore not necessarily convex.

Consider an MGUP instance with a nested optimal sequence 
\(\langle(R_1, S_1), (R_2, S_2), \ldots, (R_t, S_t)\rangle\).
Let \((R_{t+1}, S_{t+1})\) denote the solution obtained by greedily selecting one additional edge based on the $t$-optimal solution. We will prove that this solution is $(t+1)$-optimal.

Suppose that there exists a $(t+1)$-optimal solution \((R^*, S^*)\) such that there exist agents \(i\) and \(j\) with \(u_i^* < u_i\) and \(u_j^* > u_j\), where \(u_i^*\) denotes the number of upgraded edges for agent \(i\) in \(S^*\), and \(u_i\) is the number in \(S_{t+1}\).

By the greedy property of the nested optimal sequence, we know that when \(S_{t+1}\) chooses to upgrade the \((u_i^* + 1)\)-th edge for agent \(i\), the value of \(F_i\) must have been the maximum among all agents; otherwise, the greedy algorithm would have chosen an edge from a different agent. Therefore, we have:
\[
\obj(S_{t+1}) \leq F_i(\Delta^{(i)}(u_i^*)).
\]
Since \(F_i(\Delta^{(i)}(u_i^*))\) is a lower bound on \(\obj(S^*)\), we further obtain:
\[
\obj(S_{t+1}) \leq \obj(S^*),
\]
proving that the nested solution \((R_{t+1}, S_{t+1})\) is $(t+1)$-optimal.
\end{proof}

\subsection{Interval Fairness Constraints}

We consider an extension of the model in which interval fairness constraints are introduced. Specifically, for each follower agent \(i\), the number of upgraded edges is required to lie within a given range \([p_i, q_i]\). That is, the problem becomes:

\begin{equation}
\begin{aligned}
    \min_{S \subseteq E,\; |S| \le k} \quad &\sum_{i \in [n]} F_i\left( \delta_S(G^{(i)}) \right) \\
    \text{s.t.} \quad \delta_S(G^{(i)}) &= \sum_{e \in \Ti_S} c_S(e) \; \qquad \forall i \in [n]~, \\
    p_i \leq &|S \cap E^{(i)}| \leq q_i \qquad  \forall i \in [n]~.
\end{aligned}
\tag{Fair-MGUP}
\end{equation}

We w.l.o.g. assume that there exist feasible solutions under the fairness constraints, i.e., $\sum_{i\in [n]}p_i \leq k \leq \sum_{i\in [n]}p_i$.

\begin{algorithm}[tb]
\caption{Greedy Upgrading for Fair-MGUP}
\label{alg:graph_fair}
\begin{algorithmic}[1] 
\REQUIRE Connected graphs $\{G^{(i)}=(V^{(i)},E^{(i)})\}_{i\in [n]}$, two edge cost functions $\hc(\cdot), \cc(\cdot)$, functions $\{F_i\}_{i\in [n]}$, fairness constraints $\{[p_i,q_i]\}_{i\in [n]}$, and upgrade quota $k\in [ \sum_{i\in [n]}p_i, \sum_{i\in [n]} q_i ]$~. 
\ENSURE An edge subset $S$. 
\STATE For each graph $G^{(i)}$, compute $\Ti_E$ (i.e., the unique MST under the upgraded cost $\cc$). 
\STATE Initialize $S\gets \emptyset$.
\WHILE{ $|S| < k$ }
\STATE Let \(\cP:= \{i\in [n] \mid |S\cap E^{(i)}| < p_i\}\) denote the set of agents whose fairness lower bounds have not yet been satisfied.
\IF{$\cA \neq \emptyset$}
\STATE Pick an edge $e\in \bigcup_{i \in \cP} \Ti_{E}$ such that $\obj(S\cup\{e\})$ is minimized.
\ELSE 
\STATE Let \(\cQ := \{i\in [n] \mid |S\cap E^{(i)}| < q_i\}\) denote the set of agents whose fairness upper bounds have not yet been reached.
\STATE Pick an edge $e\in \bigcup_{i \in \cQ} \Ti_{E}$ such that $\obj(S\cup\{e\})$ is minimized.
\ENDIF
\STATE $S\gets S\cup \{e\}$.
\ENDWHILE
\STATE \textbf{return} $S$.
\end{algorithmic}
\end{algorithm}

\begin{theorem}\label{thm:graph_fair}
    Given any Fair-MGUP instance,~\cref{alg:graph_fair} returns an optimal subset in polynomial time.
\end{theorem}

\begin{proof}
    Clearly, all the previous results for the single-agent case still apply here, as \(k\) is guaranteed to lie within \([p, q]\), and the presence or absence of fairness constraints does not affect the single-agent setting.
For the multiagent case, the algorithm first greedily selects elements to satisfy the fairness lower bounds for all agents. Then, among agents who have not yet reached their fairness upper bounds, it continues to select elements greedily. Let \(t_0 = \sum_{i \in [n]} p_i\). By the nestedness property in the single-agent case, we know that at the end of the \(t_0\)-th iteration, the algorithm has reached a \(t_0\)-optimal solution (since it selects the \(p_i\) elements that minimize the cost for each agent \(i\)).

For any \(t > t_0\), we can apply the same argument as in the proof of~\cref{lem:multi-lq}. By leveraging the non-decreasing and convex nature of \(F_i\), we can always construct a nested \(t\)-optimal solution, implying the neatness in this fair model and completing the proof.
\end{proof}
\section{Conclusion}\label{sec:con}

In this paper, we introduce a general multiagent matroid upgrading problem that models a wide range of practical network upgrading scenarios. We prove that a simple greedy algorithm solves this problem optimally. Furthermore, our results extend naturally to more general settings, including those incorporating fairness constraints.


Our work opens several directions for future research. For example, it would be interesting to study a more general setting where the upgrade quota varies for each element. Notably, this setting generalizes the knapsack problem and is therefore NP-hard. Investigating whether greedy algorithms can still provide strong approximation guarantees in this context, or whether new algorithmic approaches are necessary, remains an interesting open problem.





\begin{acks}
This work is supported by the National Key Research and Development Program of China (2023YFA1009402, 2025YFC2423000), NSFC Programs (62302166, 62161146001, 62372176), Shanghai Key Lab of Trustworthy Computing, Shanghai Frontiers Science Center of Molecule Intelligent Syntheses and Fundamental Research Funds for the Central Universities.
\end{acks}


\bibliographystyle{ACM-Reference-Format} 
\bibliography{ref}

\newpage
\appendix

\onecolumn

\section{Multiagent Matroid Upgrading}\label{sec:matroid}

This section addresses the general multiagent matroid upgrade problem. We begin by providing some background.



\begin{definition}[Matroid]
\label{def:Matroid}
    A set system $\cI \subseteq 2^E$ defined on element set $E$ is call a \emph{matroid} if
    \begin{itemize}
        \item (Hereditary Property) given any set $S\in \cI$ and $R\subseteq S$, we have $R\in \cI$;
        \item (Exchange Property) given any two sets $S\in \cI$, $R\in \cI$ with $|S|>|R|$, there exists an element $e\in S\setminus R$ such that $R\cup \{e\}\in \cI$.
    \end{itemize}
    Denote this matroid by $(E,\cI)$.
    We call a set $S$ \emph{independent} if $S\in \cI$. A maximal independent set, that becomes dependent upon adding any element of $E$, is called a \emph{basis} of the matroid; a minimal dependent set, that becomes independent upon removing any element from the set, is called a \emph{circuit} of the matroid.
\end{definition}

We adopt the notation from earlier. Let \( \Bi_S \) denote the minimum-cost basis of matroid \( \cM^{(i)} \) under the edge cost function \( c_S \), and let \( \delta_S(\cM^{(i)}) := \sum_{e \in \Bi_S} c_S(e) \) represent its total cost. Let \( \cB_S := \bigcup_{i \in [n]} \Bi_S \) denote the union of all agents' minimum-cost bases. Without loss of generality, we assume that \( \Bi_S \) is unique for each \( S \), which can be guaranteed by applying a fixed tie-breaking rule (e.g., a consistent ordering over the elements).




\begin{algorithm}[tb]
\caption{Greedy Upgrading for MMUP}
\label{alg:matroid}
\begin{algorithmic}[1] 
 \REQUIRE A set of matroids $\{\cM^{(i)}=(E^{(i)},\cI^{(i)})\}_{i\in [n]}$, two element cost functions $\hc(\cdot), \cc(\cdot)$, upgrade quota $k$, and functions $\{F_i\}_{i\in [n]}$. 
 \ENSURE An element subset $S$. 
\STATE For each graph $G^{(i)}$, compute $\Bi_E$ (i.e., the unique min-cost basis  under the upgraded cost $\cc$). Let $\cB_E := \bigcup_{i \in [n]} \Bi_{E} $.
\STATE Initialize $S\gets \emptyset$.
\WHILE{ $|S| < k$ }
\STATE Pick an edge $e\in \cB_E$ such that $\obj(S\cup\{e\})$ is minimized.
\STATE $S\gets S\cup \{e\}$.
\ENDWHILE
\STATE \textbf{return} $S$.
\end{algorithmic}
\end{algorithm}

\begin{theorem}\label{thm:matroid}
    Given any MMUP instance,~\cref{alg:matroid} returns an optimal $S$ in polynomial time. 
\end{theorem}

We employ a similar proof strategy to that of~\cref{thm:mst} and aim to establish two technical lemmas -- \emph{inclusion property} and \emph{nestedness} in the following. 

\begin{lemma}[Matroid Inclusion Property]\label{lem:matroid_inclu}
    There always exists an optimal edge subset $S^*$ and the resulting matroid bases $\{\Bi_{S^*}\}_{i\in [n]}$ such that $\forall i\in [n]$, $\Bi_{S^*} \subseteq \Bi_\emptyset\cup \Ti_E $.  
\end{lemma}

\begin{lemma}[Matroid Nestedness]\label{lem:matroid_nest}
 Given a MMUP instance, for any nested optimal sequence 
    \(\langle(R^*_1, S^*_1), (R^*_2, S^*_2), \ldots, (R^*_t, S^*_t)\rangle\), 
    there always exists a \((t+1)\)-optimal solution \((R^*_{t+1}, S^*_{t+1})\) that preserves the nestedness property, i.e., 
    \(
    R^*_t \subseteq R^*_{t+1}\) and \( S^*_t \subseteq S^*_{t+1}~.
    \)
\end{lemma}

We leverage the following existing results in our analysis.

\begin{lemma}[Circuit Property~\cite{DBLP:books/daglib/0070636}]\label{lem:cir}
    Given a matroid $(E,\cI)$ and two circuits $C,D$, if $C\neq D$ and $e\in C\cap D$, then there exists a circuit $F \subseteq C\cup D \setminus\{e\}$. Further, this claim holds for any two dependent sets $C$ and $D$.
\end{lemma}

\begin{lemma}[Union Property~\cite{DBLP:books/daglib/0070636}]\label{lem:matroid_union}
    A union of $n$ matroids $\{(E^{(i)},\cI^{(i)})\}_{i\in [n]}$ is defined to be 
    \[\left(E=\bigcup_{i\in [n]} E^{(i)}, \cI = \left\{\bigcup_{i\in [n]} S^{(i)} \mid S^{(i)} \in \cI^{(i)} \right\} \right)~,\]
    which is also a matroid.
\end{lemma}

\begin{lemma}[Multiple Symmetric Basis Exchange Property~\cite{greene1973multiple}]\label{lem:base_ex}
    Given a matroid $(E,\cI)$ and two bases $A,B$, if $X\subseteq A\setminus B$, then there exists a subset $Y\subseteq B \setminus A$ such that both $(A\setminus X) \cup Y$ and $(B\setminus Y) \cup X$ are bases.
\end{lemma}

\begin{proof}[Proof of~\cref{lem:matroid_inclu}]
    We build on the union property (\cref{lem:matroid_union}) and the multiple symmetric basis-exchange property of matroids (\cref{lem:base_ex}) to give the proof. Consider an optimal solution $S^*$ and a resulting matroid bases union $\cB_{S^*}$. Due to~\cref{lem:matroid_union}, $\cB_{S^*}$ is a basis of the union of all $n$ matroids. 
    Let $A \subseteq \cB_{S^*} \setminus \cB_\emptyset$ denote the set of un-upgraded elements in $\cB_{S^*}$ but not present in $\cB_\emptyset$. By~\cref{lem:base_ex}, there exists a subset $B \subseteq \cB_{\emptyset} \setminus \cB_{S^*} $ such that both $(\cB_{S^*} \setminus A)\cup B$ and $(\cB_\emptyset \setminus B)\cup A$ are bases. As $\cB_\emptyset$ is a min-cost basis with respect to cost function $\hc$, we have $\hc(B) \leq \hc(A)$, implying that we obtain a new basis $(\cB_{S^*} \setminus A)\cup B$ without increasing the cost value. Similarly, we can replace all upgraded edges in $\cB_{S^*}\setminus \cB_{E}$ and complete the proof.
\end{proof}

\begin{proof}[Proof of~\cref{lem:matroid_nest}]




The new challenge in this proof lies in demonstrating that, in the single-agent case, $(R^*_1, S^*_1)$ is always contained in some $k$-optimal solution (corresponding to~\cref{lem:single_1}), as a general matroid does not have a graph structure like MSTs. The remainder of the analysis follows similar patterns to the previous one. Therefore, we omit them and focus on the analysis in the single-agent case, showing that for any $k\geq 2$, there always exists a $k$-optimal solution $(R_k^*,S_k^*)$ such that $R_1^* \subseteq R_k^*$ and $S_1^* \subseteq S_k^*$.  

Consider a $1$-optimal solution $(R_1^*=\{r\},S_1^*=\{s\})$ and a $k$-optimal solution $(R_k^*=\{r_1,\ldots,r_k\},S_k^*=\{s_1,\ldots,s_k\})$ that does not satisfy the nestedness. We distinguish three cases:
\begin{enumerate}[itemindent=2em, label=(\arabic*)]
        \item $r \notin R_k^* $ and $s \in S_k^* $~,
        \item $r \in R_k^* $ and $s \notin S_k^* $~,
        \item $r \notin R_k^* $ and $s \notin S_k^* $~.
    \end{enumerate}

We say an element $e$ (resp. an element set $A$) is \emph{exchangeable} to an element $e'$ (resp. an element set $A'$) on an independent set $B$ if $(B\setminus \{e'\}) \cup \{e\}$ (resp. $(B\setminus A') \cup A$) is independent. 

For case (1), assume w.l.o.g. that $s=s_1$.
We prove that there must exist an element $r'\in R_k^*$ that is exchangeable to $r$ on $P\cup S_k^*$ and has $\hc(r') \leq \hc(r)$, where $P= \B_{\emptyset}\setminus R_k^*$ and $P\cup S_k^*$ represents the matroid basis obtained by swapping $R_k^*$ in $\B_{\emptyset}$ with $S_k^*$. Let $X \subseteq R_k^*$ be the set of elements in $R_k^*$ to which $s$ is exchangeable on $\B_{\emptyset}$. This set is guaranteed to be non-empty according to~\cref{lem:base_ex}.
Due to the optimality of $(r,s)$, we know that $\hc(r_i) \leq \hc(r) $ for any $r_i \in X$. Consequently, it suffices to show that $\exists r_i \in X$ is exchangeable to $r$ on $P\cup S_k^*$.

By the definition of $X$, adding $s$ into $\B_\emptyset$ produces a circuit $L$ that satisfies $L\cap R_k^* = X$. Assume for contradiction that any $r_i \in X$ is not exchangeable to $r$ on $P\cup S_k^*$. This implies that adding each $r_i\in X$ into $P\cup S_k^*$ forms a circuit $L_i$ that contains $r_i$. Applying the circuit property (\cref{lem:cir}) to circuits $L$ and each $L_i$, we conclude that a circuit exists in element subset $P\cup S_k^*$, contradicting the fact that $P\cup S_k^*$ is the basis in $k$-optimal solution $(R_k^*,S_k^*)$.




For case (2), assume w.l.o.g. that $r=r_1$. We prove that there exists an element $s' \in S_k^*$ to which $s$ is exchangeable on $P\cap S_k^*$ and has $\cc(s')\geq \cc(s)$. Let $Y$ be the set of elements in $S_k^*$ to which $s$ is exchangeable on $P\cap S_k^*$. Again,~\cref{lem:base_ex} guarantees that $Y\neq \emptyset$. If we demonstrate that $\exists s_i \in Y$ is exchangeable to $r$ on $\B_{\emptyset}$, then by the optimality of $(r,s)$, we have $\cc(s_i)\geq \cc(s)$ and complete the proof of this case. 

Assume for contradiction that no element in $Y$ is exchangeable to $r$ on $\B_{\emptyset}$. Consequently, adding each $r_i\in Y$ into $\B_\emptyset$ forms a circuit $L_i$ that does not contain element $r$. By the definition of $Y$, adding $s$ into $P\cap S_k^*$ produces a circuit $L$ that satisfies $L\cap S_k^* = Y$. Similarly, we apply the circuit property (\cref{lem:cir}) to circuits $L$ and each $L_i$ and see that a circuit exists in subset $ (\B_{\emptyset} \setminus \{r\}) \cup \{s\} $, contradicting the fact that this subset is the basis in $1$-optimal solution $(R_1^*,S_1^*)$.

The last case, where $r\notin R_k^*$ and $s\notin S_k^*$ is particularly challenging. We will build upon the analyses of the two previous cases and apply the exchange property multiple times to tackle this scenario. Let $P:=\B_{\emptyset} \setminus (\{r\} \cup R_k^*)$. Then the bases in the $1$-optimal solution and $k$-optimal solution can be denoted by $ P\cup R_k^* \cup \{ s \} $ and $ P \cup S_k^* \cup \{ r \} $, respectively.
By applying~\cref{lem:base_ex} to bases $ P\cup R_k^* \cup \{ s \} $ and $\B_\emptyset$, we can w.l.o.g. assume that $ P \cup \{r,r_1\} \cup ( S_k^* \setminus \{ s_1 \} ) $ is a matroid basis and $s_1$ is exchangeable to $r_1$ on $\B_\emptyset$, implying that $\hc(r_1) - \cc(s_1) \leq \hc(r) - \cc(s)$ due to the optimality of $(r,s)$. Now if $s$ is exchangeable to $r$ on $ P \cup \{r,r_1\} \cup ( S_k^* \setminus \{ s_1 \} )$, the proof has been completed as we obtain a nested basis without increasing the cost. 

Consider the scenario where $s$ is not exchangeable to $r$ on $ P \cup \{r,r_1\} \cup ( S_k^* \setminus \{ s_1 \} )$. By applying~\cref{lem:base_ex} sequentially among $ P \cup \{r,r_1\} \cup ( S_k^* \setminus \{ s_1 \} )$, $ P\cup R_k^* \cup \{ s \} $ and $ P \cup S_k^* \cup \{ r \} $, we are able to show that there exists an element $ r_t \in  R_k^* $ that is exchangeable to $r$ on $ P \cup S_k^* \cup \{ r \}  $, thereby obtaining a basis $ P \cup S_k^* \cup \{ r_t \}  $.

We further distinguish two subcases based on whether $s$ is exchangeable to $r$ on $ P \cup S_k^* \cup \{ r \} $ or not, and reduce these two subcase to the previous cases, respectively.

For the first exchangeable subcase, we see that $ P \cup S_k^* \cup \{ s \} $ is a basis.
Let $A$ be the set of elements in $ R_k^*$ to which $s$ is not exchangeable on $\B_\emptyset$. If $A\neq \emptyset$, applying~\cref{lem:base_ex} to bases $ P \cup S_k^* \cup \{ s \} $ and $\B_\emptyset$ should give a nested optimal basis. On the other hand, if $A = \emptyset$, due to the optimality of $(r,s)$, we have $\hc(r_i) \leq \hc(r)$ for any $r_i \in R_k^*$. Recollect $r_t$ and the basis $ P \cup S_k^* \cup \{ r_t \}  $ constructed above. Since $\hc(r_i) \leq \hc(r)$, $ P \cup S_k^* \cup \{ r_t \}  $ is also a $k$-optimal basis. Now the subcase has been reduced to case (2) and the proof can be completed.  

For the second unexchangeable subcase, we see that adding $s$ into $ P \cup S_k^* \cup \{ r \} $ forms a circuit $L$ that does not contain $e$. Define $X := L \cap S_k^* $.~\cref{lem:cir} guarantees that there exists an element $s_t \in X$ that is exchangeable to $r$ on $\B_{\emptyset}$ and therefore, has $\cc(s_t) \geq \cc(s) $; otherwise, we can readily derive a contradiction that a circuit exists in a basis. By the definition of a circuit, $ P \cup (S_k^* \setminus \{s_t\}) \cup \{ r,s \} $ is a $k$-optimal basis. Finally, the subcase has been reduced to case (3) and the proof can be completed. 
\end{proof}



\section{Further Extensions for MMUP}\label{sec:exten}


In this section, we discuss two further extensions of the general multiagent matroid upgrading problem and show that the greedy algorithm remains optimal in these extended models. The overall analysis and algorithm are similar to those in~\cref{sec:graph_exten}, but for completeness, we provide them here in full.

\subsection{Minimax Objective}

We consider a variant of the model where the objective changes from the sum of the \( F_i \)'s to their maximum. That is, the problem becomes:
\begin{equation}
\begin{aligned}
    \min_{S \subseteq E,\, |S| \le k} &\max_{i \in [n]} F_i\left( \delta_S(\cM^{(i)}) \right) \\
    \text{s.t.} \quad \delta_S(\cM^{(i)}) = &\min_{B_i \in \cB^{(i)}} \sum_{e \in B_i} c_S(e) \quad \forall i \in [n]~,
\end{aligned}
\tag{Minimax-MMUP}
\end{equation}

\begin{theorem}\label{thm:minmax}
    Given any Minimax-MMUP instance,~\cref{alg:matroid} returns an optimal subset in polynomial time.
\end{theorem}

\begin{proof}
    The proof is nearly identical to that of~\cref{thm:matroid}; the only difference lies in the final part of the nestedness proof for multiple agents (corresponding to~\cref{lem:multi-lq}). All other lemmas, including those for the single-agent case, remain applicable. The difference here arises because the objective is now a maximum function rather than a sum, and is therefore not necessarily convex.

Consider an MMUP instance with a nested optimal sequence 
\(\langle(R_1, S_1), (R_2, S_2), \ldots, (R_t, S_t)\rangle\).
Let \((R_{t+1}, S_{t+1})\) denote the solution obtained by greedily selecting one additional element based on the $t$-optimal solution. We will prove that this solution is $(t+1)$-optimal.

Suppose that there exists a $(t+1)$-optimal solution \((R^*, S^*)\) such that there exist agents \(i\) and \(j\) with \(u_i^* < u_i\) and \(u_j^* > u_j\), where \(u_i^*\) denotes the number of upgraded elements for agent \(i\) in \(S^*\), and \(u_i\) is the number in \(S_{t+1}\).

By the greedy property of the nested optimal sequence, we know that when \(S_{t+1}\) chooses to upgrade the \((u_i^* + 1)\)-th element for agent \(i\), the value of \(F_i\) must have been the maximum among all agents; otherwise, the greedy algorithm would have chosen an element from a different agent. Therefore, we have:
\[
\obj(S_{t+1}) \leq F_i(\Delta^{(i)}(u_i^*)).
\]
Since \(F_i(\Delta^{(i)}(u_i^*))\) is a lower bound on \(\obj(S^*)\), we further obtain:
\[
\obj(S_{t+1}) \leq \obj(S^*),
\]
which proves that the nested solution \((R_{t+1}, S_{t+1})\) is indeed $(t+1)$-optimal.
\end{proof}

\subsection{Interval Fairness Constraints}

We consider an extension of the model in which fairness constraints are introduced. Specifically, for each follower agent \(i\), the number of upgraded elements is required to lie within a given range \([p_i, q_i]\). That is, the problem becomes:

\begin{equation}
\begin{aligned}
    \min_{S \subseteq E,\; |S| \le k} \quad &\sum_{i \in [n]} F_i\left( \delta_S(\cM^{(i)}) \right) \\
    \text{s.t.} \quad \delta_S(\cM^{(i)}) &= \min_{B_i \in \cB^{(i)}} \sum_{e \in B_i} c_S(e) \quad \forall i \in [n], \\
    p_i &\leq |S \cap E^{(i)}| \leq q_i \qquad  \forall i \in [n]~.
\end{aligned}
\tag{Fair-MMUP}
\end{equation}

Without loss of generality, we can assume that there exist feasile solutions under the given fairness constraints, i.e., $\sum_{i\in [n]}p_i \leq k \leq \sum_{i\in [n]}p_i$.

\begin{algorithm}[tb]
\caption{Greedy Upgrading for Fair-MMUP}
\label{alg:fair}
\begin{algorithmic}[1] 
 \REQUIRE A set of matroids $\{\cM^{(i)}=(E^{(i)},\cI^{(i)})\}_{i\in [n]}$, two element cost functions $\hc(\cdot), \cc(\cdot)$, functions $\{F_i\}_{i\in [n]}$, fairness constraints $\{[p_i,q_i]\}_{i\in [n]}$, and upgrade quota $k \in [ \sum_{i\in [n]}p_i, \sum_{i\in [n]} q_i ]$~. 
 \ENSURE An element subset $S$. 
\STATE For each matroid $\cM^{(i)}$, compute $\Bi_E$ (i.e., the unique min-cost basis under the upgraded cost $\cc$). 
\STATE Initialize $S\gets \emptyset$.
\WHILE{ $|S| < k$ }
\STATE Let \(\cP:= \{i\in [n] \mid |S\cap E^{(i)}| < p_i\}\) denote the set of agents whose fairness lower bounds have not yet been satisfied.
\IF{$\cA \neq \emptyset$}
\STATE Pick an element $e\in \bigcup_{i \in \cP} \Bi_{E}$ such that $\obj(S\cup\{e\})$ is minimized.
\ELSE 
\STATE Let \(\cQ := \{i\in [n] \mid |S\cap E^{(i)}| < q_i\}\) denote the set of agents whose fairness upper bounds have not yet been reached.
\STATE Pick an element $e\in \bigcup_{i \in \cQ} \Bi_{E}$ such that $\obj(S\cup\{e\})$ is minimized.
\ENDIF
\STATE $S\gets S\cup \{e\}$.
\ENDWHILE
\STATE \textbf{return} $S$.
\end{algorithmic}
\end{algorithm}

\begin{theorem}\label{thm:fair}
    Given any Fair-MMUP instance,~\cref{alg:fair} returns an optimal subset in polynomial time.
\end{theorem}

\begin{proof}
    Clearly, all the previous results for the single-agent case still apply here, as \(k\) is guaranteed to lie within \([p, q]\), and the presence or absence of fairness constraints does not affect the single-agent setting.

For the multiagent case, the algorithm first greedily selects elements to satisfy the fairness lower bounds for all agents. Then, among agents who have not yet reached their fairness upper bounds, it continues to select elements greedily. Let \(t_0 = \sum_{i \in [n]} p_i\). By the nestedness property in the single-agent case, we know that at the end of the \(t_0\)-th iteration, the algorithm has reached a \(t_0\)-optimal solution (since it selects the \(p_i\) elements that minimize the cost for each agent \(i\)).

For any \(t > t_0\), we can apply the same argument as in the proof of~\cref{lem:multi-lq}. By leveraging the non-decreasing and convex nature of the functions \(F_i\), we can always construct a nested \(t\)-optimal solution, implying the neatness in this fair model and completing the proof.
\end{proof}
\section{Connection to Budget-Constrained MST}\label{sec:connection}

In this section, we show that our model captures the special case of the budget-constrained minimum spanning tree problem (BCMP) with \(\{0,1\}\) edge weights. We begin by presenting the formal definition of the problem and then show that it is equivalent to the single-agent MGUP, where the function \(F\) can, w.l.o.g., be assumed to be the identity function.

\paragraph{The Budget-Constrained MST Problem.} In this problem, we are given an undirected graph \( G(V, E) \), where each edge \( e \in E \) has an associated cost \( c_e \) and weight \( w_e \). The objective is to find a spanning tree of minimum total cost such that the total weight does not exceed a given budget \( B \). When each edge weight \( w_e \) is either \( 0 \) or \( 1 \), we refer to this special case as the \(\{0,1\}\)-weight BCMP.

\begin{theorem}\label{eq:}
    The \(\{0,1\}\)-weight BCMP is equivalent to the single-agent MGUP, and therefore, can be solved polynomially. 
\end{theorem}

\begin{proof}
    To see the equivalence, consider a single-agent MGUP instance where each edge has an original cost \( \hc_e \) and an upgraded cost \( \cc_e \), with an upgrade quota of \( k \). We can construct an equivalent \(\{0,1\}\)-weight BCMP instance by replacing each edge with two copies: one with (weight, cost) \( (0, \hc_e) \), and the other with \( (1, \cc_e) \), and setting the budget to \( k \). The constraint of upgrading at most \( k \) edges in the original MGUP instance is equivalent to the constraint of selecting edges with total weight at most \( k \) in the BCMP instance.

Conversely, given any \(\{0,1\}\)-weight BCMP instance where each edge has a (weight, cost) pair \( (w_e, c_e) \) and a budget \( k \), we can construct an equivalent MGUP instance as follows: for each edge with \( w_e = 1 \), create an edge with \( (\hc_e = \infty,\, \cc_e = c_e) \); for each edge with \( w_e = 0 \), create an edge with \( (\hc_e = c_e,\, \cc_e = c_e) \); and set the upgrade quota to \( k \). Any feasible edge subset in the BCMP instance corresponds to a feasible solution in the MGUP instance by upgrading precisely those edges with weight 1, and the objectives remain identical.
\end{proof}
\end{document}